\documentclass[12pt]{article}

\usepackage{graphicx, amssymb, latexsym, amsfonts, amsmath, lscape, amscd,
amsthm, color, epsfig, mathrsfs, tikz, enumerate, multirow}
\usepackage{tikz,pgf,tkz-graph}
\usetikzlibrary{shapes.geometric,calc,shapes,fit,intersections,graphs,graphs.standard,decorations.pathreplacing,decorations.markings,positioning,arrows,shadows,arrows.meta,fadings,decorations.text}
\usetikzlibrary{arrows,automata}
\usetikzlibrary{graphs}
\usetikzlibrary{graphs.standard}
\tikzset{every state/.style={minimum size=0pt}}
\usepackage{booktabs}
\usetikzlibrary{hobby}
\usepackage[]{mdframed}
\tikzset{XOR/.style={draw,circle,append after command={
        [shorten >=\pgflinewidth, shorten <=\pgflinewidth,]
        (\tikzlastnode.north) edge (\tikzlastnode.south)
        (\tikzlastnode.east) edge (\tikzlastnode.west)
        }
    }
}

\usepackage[textsize=footnotesize,color=green!40]{todonotes}
\newtheorem{theorem}{Theorem}[section]

\newtheorem{corollary}[theorem]{Corollary}

\newtheorem{lemma}[theorem]{Lemma}
\newtheorem{proposition}[theorem]{Proposition}

\theoremstyle{definition}
\newtheorem{definition}[theorem]{Definition}

\theoremstyle{definition}
\newtheorem{remark}[theorem]{Remark}

\newcommand\DELETE[1]{}

\newcommand{\xor}{\oplus}
\usetikzlibrary{decorations.pathmorphing}

\usepackage{todonotes}
\usepackage{mathtools}

\begin{document}

%\begin{frontmatter}

\title{\bf Bounds and extremal graphs for monitoring edge-geodetic sets in graphs}
\author{{\sc Florent Foucaud}$\,^{a}$, {\sc Clara Marcille}$\,^{b}$, \\{\sc 
Zin Mar Myint}$\,^{c}$, {\sc 
R. B. Sandeep}$\,^{c}$, {\sc 
Sagnik Sen}$\,^{c}$, {\sc 
S Taruni}$\,^{c}$\\
\mbox{}\\
{\small $(a)$ Université Clermont Auvergne, CNRS, Clermont Auvergne INP,} \\ {\small Mines Saint-\'Etienne LIMOS, 63000 Clermont-Ferrand, France.}\\
{\small $(b)$ Univ. Bordeaux, CNRS,  Bordeaux INP, LaBRI, UMR 5800,}\\ {\small F-33400, Talence, France.}\\
{\small $(c)$ Indian Institute of Technology Dharwad, India.}\\
}

\date{\today}

\maketitle

\begin{abstract}
A monitoring edge-geodetic set, or simply an MEG-set, of a graph $G$ is a vertex subset $M \subseteq V(G)$ such that given any edge $e$ of $G$, $e$ lies on every shortest $u$-$v$ path of $G$, for some $u,v \in M$. 
The monitoring edge-geodetic number of $G$, denoted by $meg(G)$, is the minimum cardinality of such an MEG-set. 
This notion provides a graph theoretic model of the network monitoring problem. 

In this article, we compare $meg(G)$ with some other graph theoretic parameters stemming from the network monitoring problem and provide examples of graphs having prescribed values for each of these parameters. We also characterize graphs $G$ that have $V(G)$ as their minimum MEG-set, which settles an open problem due to Foucaud \textit{et al.} (CALDAM 2023), and prove that some classes of graphs fall within this characterization. We also provide a general upper bound for $meg(G)$ for sparse graphs in terms of their girth, and later refine the upper bound using the chromatic number of $G$. We examine the change in $meg(G)$ with respect to two fundamental graph operations: clique-sum and subdivisions. In both cases, we provide a lower and an upper bound of the
possible amount of changes and provide (almost) tight examples. 

\end{abstract}


\section{Introduction}
In the field of network monitoring, the networking components are monitored for faults and evaluated to maintain and optimize their availability. In order to detect failures, one of the popular methods for such monitoring processes involves setting up distance probes \cite{bampas2015network,beerliova2006network,bejerano2003robust,foucaud2022monitoring}. At any given time, a distance probe can measure the distance to any other probe in the network. If there is any failure in the connection, then the probes should be able to detect it as there would be a change in the distances between the components. Such networks can be modeled by graphs whose vertices represent the components and the edges represent the connections between them. We select a subset of vertices of the graph and call them probes. This concept of probes that can measure distances in graphs has many real-life applications, for example, it is useful in the fundamental task of routing \cite{dall2006exploring,govindan2000heuristics}, or using path-oriented tools to monitor IP networks \cite{bejerano2003robust}, or problems concerning network verification \cite{bampas2015network,beerliova2006network,bilo2010discovery}. Based on the requirements of the networks, there have been various related parameters that were defined on graphs in order to study the problem and come up with an effective solution. To name a few, we may mention 
the geodetic number~\cite{chartrand2002geodetic,chartrand2002geodeticsurvey,dourado2010some,harary1993geodetic},
the edge-geodetic number~\cite{atici2003edge,santhakumaran2007edge},
the strong edge-geodetic number~\cite{irvsivc2018strong,MKXAT17},
and the distance-edge monitoring number~\cite{foucaud2022monitoring}.
The focus of this article is on studying the \textit{monitoring edge-geodetic number} of a graph, a concept related to the above ones and introduced in~\cite{foucaud2023monitoring} (see also~\cite{foucaud2023monitoringfull,haslegrave2023monitoring}).

\subsection{Main definition}\label{subsec:defs}
We deal with simple graphs, unless otherwise stated, and will use standard graph terminology and notation according to West~\cite{west2001introduction}. 

\begin{definition}
Let $G$ be a graph. A pair of vertices $u,v \in V(G)$, or any vertex subset $M \subseteq V(G)$ containing them, \textit{monitors} an edge $e \in E(G)$ if $e$ belongs to all shortest paths between $u$ and $v$.  
A \emph{monitoring edge-geodetic set} or simply, an \emph{MEG-set} of $G$  is a vertex subset $M \subseteq V(G)$ that monitors every edge of $G$. 
 The \emph{monitoring edge-geodetic number}, denoted by $meg(G)$, is the 
 cardinality of a minimum MEG-set of $G$. 
\end{definition}

The notion of MEG-sets were introduced in~\cite{foucaud2023monitoring}, motivated by the above network monitoring application: the vertices of the MEG-set are distance-probes, that can measure the distance between each other. By the definition of an MEG-set, if some edge of the network fails, then there is at least one pair of probes whose distance changes, and so, the system of probes is able to detect the failure.

\subsection{Previous works}
In the paper that introduced MEG-sets~\cite{foucaud2023monitoring}, the value of $meg(G)$, when $G$ belongs to some basic graph families, were determined. These graph families include the family of trees, unicyclic graphs, complete graphs, 
complete multipartite graphs, rectangular grids, and hypercubes. The authors of~\cite{foucaud2023monitoring} also showed a relation with the feedback edge set number $f$ and the number $\ell$ of leaves of the graph: $meg(G)\leq 9f+\ell-8$ (for $f\geq 2$). This was improved to $meg(G)\leq 3f+\ell+1$ (which is tight up to an additive factor of 1) in~\cite{CFH23}. In~\cite{haslegrave2023monitoring}, $meg(G)$ when $G$ is the result of certain types of graph products was studied, in particular, Cartesian products and strong products. The case when $G$ is a corona product was studied in~\cite{foucaud2023monitoringfull}. In~\cite{haslegrave2023monitoring}, it was shown that determining $meg(G)$ is an NP-complete problem. This was refined to hold for graphs of maximum degree at most~9 in~\cite{foucaud2023monitoringfull}.

\subsection{Related parameters}
There are some network monitoring based graph parameters studied in the literature, whose definitions are relevant in our context since they are related to $meg(\cdot)$ (see Section~\ref{sec gap}).  We list them below.

\begin{itemize}
    \item[-] A \textit{geodetic set} of a graph $G$ is a vertex subset $S \subseteq V(G)$ such that every vertex of $G$ lies on some shortest path between 
    two vertices $u,v \in S$. The \textit{geodetic number}, denoted by $g(G)$, is the minimum $|S|$, where $S$ is a geodetic set of $G$. The concept was introduced by Harary et al. in 1993~\cite{harary1993geodetic} and received considerable attention since then, both from the structural side~\cite{chartrand2002geodetic,chartrand2002geodeticsurvey,dourado2010some} and from the algorithmic side~\cite{CDFGLR20,KK22}.

    \item[-] An \textit{edge-geodetic set} of a graph $G$ is a vertex subset $S \subseteq V(G)$ such that every edge of $G$ lies on some shortest path between 
    two vertices $u,v \in S$. The \textit{edge-geodetic number}, denoted by $eg(G)$, is the minimum $|S|$, where $S$ is an edge-geodetic set of $G$. This was introduced in 2003 by Atici et al.~\cite{atici2003edge} and further studied from the structural angle~\cite{santhakumaran2007edge} as well as algorithmic angle~\cite{CGR23,DIT21}.

    \item[-] A \textit{strong edge-geodetic set} of a graph $G$ is a vertex subset $S \subseteq V(G)$ and an assignment of a particular shortest $u$-$v$ path $P_{uv}$ to each pair of distinct vertices $u,v \in S$ such that every edge of $G$ lies on $P_{uv}$ for some $u,v \in S$. The \textit{strong edge-geodetic number}, denoted by $seg(G)$, is the minimum $|S|$, where $S$ is a strong edge-geodetic set of $G$. This concept was introduced in 2017 by Manuel et al.~\cite{MKXAT17}. See~\cite{irvsivc2018strong} for some structural studies, and~\cite{DIT21} for some algorithmic results.

     \item[-] A \textit{distance edge-monitoring set} of a graph $G$ is a vertex subset $S \subseteq V(G)$ such that for every edge $e$ of $G$, there are two vertices $x\in S$ and $y\in G$ for which $e$ lies on all shortest paths between $x$ and $y$. Let $dem(G)$ be the smallest size of a 
     distance edge-monitoring set
     of $G$. 
     This concept was introduced in 2020 by Foucaud et al.~\cite{foucaud2022monitoring}. Refer to~\cite{foucaud2022monitoring, foucaud2020dem} for both structural and algorithmic results for this parameter.
\end{itemize}

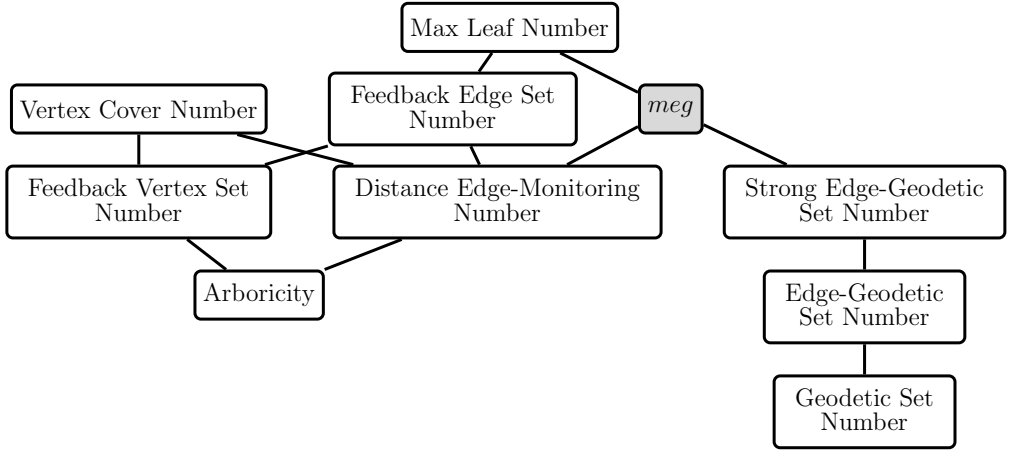
\begin{figure}[ht!]
\centering
\tikzstyle{mybox}=[line width=0.5mm,rectangle, minimum height=.8cm,fill=white!70,rounded corners=1mm,draw]
\tikzstyle{myedge}=[line width=0.5mm]
\newcommand{\tworows}[2]{\begin{tabular}{c}{#1}\\[-1mm]{#2}\end{tabular}}
\scalebox{0.8}{\begin{tikzpicture}[node distance=10mm]
 		\node[mybox] (vc)  {Vertex Cover Number};
 		 \node[mybox] (mln) [above right = of vc,xshift=15mm,yshift=-2mm] {Max Leaf Number};
  		\node[mybox] (fes) [right = of vc] {\tworows{Feedback Edge Set}{Number}}  edge[myedge] (mln);
  		 \node[mybox] (meg) [right = of fes,fill=black!15] {$meg$} edge[myedge] (mln);
 		\node[mybox] (fvs) [below = of vc, yshift=5mm] {\tworows{Feedback Vertex Set}{Number}} edge[myedge] (vc) edge[myedge] (fes);
% 		\node[mybox] (td) [left = of fvs] {Treedepth} edge[myedge] (vc);
% 		\node[mybox] (bw) [left = of td] {Bandwidth};
% 		\node[mybox] (pw) [below = of td] {Pathwidth} edge[myedge] (td) edge[myedge] (bw);
 		\node[mybox] (dem) [right = of fvs] {\tworows{Distance Edge-Monitoring}{Number}} edge[myedge] (fes) edge[myedge] (vc) edge[myedge] (meg);
 		 \node[mybox] (segs) [right = of dem] {\tworows{Strong Edge-Geodetic}{Set Number}} edge[myedge] (meg);
 		  \node[mybox] (egs) [below = of segs, yshift=5mm] {\tworows{Edge-Geodetic}{Set Number}} edge[myedge] (segs);
 		  \node[mybox] (gs) [below = of egs, yshift=5mm] {\tworows{Geodetic Set}{Number}} edge[myedge] (egs);
 		\node[mybox] (arb) [below = of fvs,xshift=20mm, yshift=5mm] {Arboricity} edge[myedge] (fvs) edge[myedge] (dem);%edge[myedge] (pw) 
% 		\node[mybox] (c) [below = of arb] {Clique Number} edge[myedge] (arb);

\end{tikzpicture}}
\caption{Relations between the parameter $meg$ and other structural parameters in graphs (with no isolated vertices). For the relationships of distance edge-monitoring sets, see~\cite{foucaud2022monitoring}. Edges between parameters indicate that the value of the bottom parameter is upper-bounded by a function of the top parameter. Figure reproduced from~\cite{foucaud2023monitoringfull}.}
\label{fig:diagram}
\end{figure}

Apart from the above, some well-known graph parameters too have relation with 
$meg(\cdot)$. 
See Figure~\ref{fig:diagram} (reproduced from~\cite{foucaud2023monitoringfull}) for a diagram showing the relations between these parameters, and others.

 \subsection{Our results, and organization of the paper} 
 
 \begin{itemize}
     \item In Section~\ref{sec gap}, we
     recall that 
     $g(G) \leq eg(G) \leq seg(G) \leq meg(G)$ for any connected non-trivial graph $G$~\cite{foucaud2022monitoring}.
     To complement our findings, we  construct examples of graphs 
     $G_{a,b,c,d}$ having $g(G) = a, eg(G) = b, seg(G) =c$, and $meg(G) = d$, where $a \leq b \leq c \leq d$ (and a few additional constraints). 
     We also show that $dem(G) < meg(G)$, and construct examples of graphs $G_{p,q}$ having $dem(G) = p$ and $meg(G) = q$, where $1 \leq p < q$.

     \item In Section~\ref{sec megfull} we prove a necessary and sufficient
     condition for a vertex to be part of every MEG-set of $G$. 
     As a corollary, we characterize graphs $G$ having $meg(G) = |V(G)|$, which answers an open question posed by Foucaud, Krishna and Ramasubramony Sulochana~\cite{foucaud2023monitoring}. It was worth recalling that this open question was addressed (not solved) by Haslegrave~\cite{haslegrave2023monitoring}. 
     % Additionally, we also prove a sufficient condition of when a vertex is never part of any minimum MEG-set of the graph. 
     These results can act as fundamental tools 
     in the study of MEG-set and related problems.

    \item In Section~\ref{sec:ext}, we will explore the impact of the tools built in Section~\ref{sec megfull} to expand the list of known graphs whose minimum MEG-set is the entire vertex set (defined as \emph{MEG-extremal} graphs in Section~\ref{sec:ext}). 
    To be precise, we completely characterize 
    $meg(G)$ when $G$ is 
    a cograph, 
    a block graph, 
    a well-partitioned chordal graph,
    or a proper interval graph, and observe that the 2-connected graphs from these families are MEG-extremal.   
    Moreover, if $G$ is an MEG-extremal graph, then the Cartesian and the strong products
    of $G$ with another graph is also MEG-extremal. The former was already known due to 
    Haslegrave~\cite{haslegrave2023monitoring},
    we however provide a shorter proof. 
    Furthermore, we also show that the tensor product of two MEG-extremal graphs is MEG-extremal.

    \item In Section~\ref{sec graph parameters}, we show that 
 $meg(G) \leq \frac{4|V(G)|}{g-3}$, where $g \geq 4$ denotes the girth of $G$. As a consequence, we show that any vertex cover is an MEG-set for a graph having girth at least $5$. Later we also prove a refinement of the main result of this section using the chromatic number of $G$.

     \item  In Section~\ref{sec operations}, we explore the effect of two fundamental graph operations, namely, the clique-sum and the subdivision on $meg(G)$. We show that $meg(G)$ is both lower and upper bounded by functions related to the operations and that the bounds are almost tight.

     \item  In Section \ref{sec conclusions}, we share our concluding remarks which also contain suggestions for future works in this direction. 
 \end{itemize}
 
 \medskip
  
\noindent \textbf{Note:} A preliminary version of this article was published in the proceedings of the CALDAM 2024 conference~\cite{FMMSST24}. 
This version contains an extended introduction, proofs missing from the conference paper, extensions of the results, additional results, and a corrected version of Theorem~\ref{thm:large_girth}. However, the algorithmic results from~\cite{FMMSST24} are not contained in the current article, but presented in a separate extended version~\cite{algo-MEG}.

  \section{Preliminary results}\label{sec:preliminary}
Right before we proceed with our contributions, let us recall a few results from~\cite{foucaud2023monitoring} since we think they will give a proper initial insight into the notion of MEG-set, and will also be used on several occasions throughout this work.

    A vertex is \textit{simplicial} if its neighborhood forms a clique.    A vertex $v$ of a graph $G$ is \textit{pendent vertex} if it is of degree~1.

\begin{lemma}[Lemma 2.1 of~\cite{foucaud2023monitoring}]\label{lemma:simpl}
    In a connected graph $G$ with at least one edge, any simplicial vertex belongs to any edge-geodetic set and thus, to any MEG-set of $G$.
\end{lemma}

Lemma \ref{lemma:simpl} implies that all pendent vertices of a graph $G$ are part of any MEG-set of $G$.

    Two distinct vertices $u$ and $v$ are said to be \textit{open twins} if $N(u) = N(v)$ and \textit{closed twins} if $N[u] = N[v]$. If $u, v$ are either closed or open twins, then we simply call them \textit{twins}. 

\begin{lemma}[Lemma 2.2 of~\cite{foucaud2023monitoring}]\label{lemma:twins}
    If two vertices are twins of degree at least $1$ in a graph $G$, then they must belong to any MEG-set of $G$.
\end{lemma}

% \todo[inline]{F: add needed definitions here?}

\section{Relation between network monitoring parameters}\label{sec gap}
In this section, we find examples of graphs having prescribed values of network monitoring parameters.

\subsection{Relation with geodetic parameters}
We start with proving an if and only if condition for which a strong edge-geodetic set is also an MEG-set.

\begin{proposition}
    Let $M \subseteq V(G)$ be a vertex subset of a graph $G$ and let $f$ be an assignment of a shortest path 
    to each pair of vertices of $M$.  Then $M$ is an MEG-set if and only if $M$, along with the assignment $f$, is a strong edge-geodetic set for any choice of $f$. 
\end{proposition}

\begin{proof}
  Let $M$ be an MEG-set of $G$. 
  As all the edges of $G$ are monitored by $M$, 
  for each edge $e$ there exists a pair of 
  vertices $u,v \in M$ such that $e$ lies on every 
  shortest path between $u$ and $v$. 
  Therefore, if we arbitrarily assign any shortest path to each pairs of vertices of $M$, the set $M$ along with this so-obtained assignment of shortest paths 
  is a strong edge-geodetic set. 

  On the other hand, let $M$, along with the assignment $f$, be a strong edge-geodetic set for any choice of $f$. Now if some edge $e$ of $G$ is not monitored, then it is possible to find an assignment of  a path $P_{uv}$ to each pair $u,v$ of vertices of $M$ 
  such that $e$ is not contained in $P_{uv}$ for any $u,v$. This is a contradiction to our assumption. Thus, $M$ must be an MEG-set. 
\end{proof}

It is known~\cite{foucaud2023monitoring} (see also Figure~\ref{fig:diagram}) that the following relation holds among the geodetic parameters: 
$$g(G) \leq eg(G) \leq seg(G) \leq meg(G).$$
Moreover, we know that for any connected graph $G \neq K_1$, $g(G) \geq 2$. 
Therefore,  it is natural to ask the question, given four positive integers $a,b,c,d$  satisfying 
 $2 \leq a \leq b \leq c \leq d$, is there a graph $G_{a,b,c,d}$ such that we have $g(G_{a,b,c,d}) = a$,  $eg(G_{a,b,c,d}) = b$,  $seg(G_{a,b,c,d}) = c$, and $meg(G_{a,b,c,d}) = d$? The following remark captures some basic version of this answer. 
 Later, in Theorem~\ref{thm Gabcd}  we provide a positive answer to this question except for some specific cases.

\begin{remark}
Notice that, for any complete graph $K_n$ on $n\geq 2$ vertices, the values of all the parameters are equal to $n$. That is, equality holds in all the inequalities of the above chain of inequations. 

On the other hand, Figure~\ref{fig example reln} gives an example of a graph where all the inequalities of the above chain of inequations are strict. To be specific, in this particular example, the values of the parameters increase exactly by one in each step.
%\todo{F: names of the vertices are slightly different in the figure and caption}

\end{remark}
% \SenLine{I prefer the vertex names to be $v_1, v_2, v_3, v_4, v_5$ in the figure instead of numbers. That will be consistent with our style of naming conventions used later.}

    \begin{figure}[h]
\centering
    \begin{tikzpicture}[scale=0.7, transform shape]
    \node[minimum size=4cm,draw,regular polygon,regular polygon sides=5] (a) {};
\foreach \i in {1,...,5}
    \node[circle, draw,inner sep=0.1mm,minimum size=3mm, fill] (a\i) at (a.corner \i) {$i$};
   
  \foreach \x in {1,2}
        \foreach \y in {3,5}
\draw (a\y) -- (a\x);
\node at (-2.3,0.7) {$v_2$};
\node at (-1.5,-1.7) {$v_3$};
\node at (1.5,-1.7) {$v_4$};
\node at (2.3,0.7) {$v_5$};
\node at (0.3,2.2) {$v_1$};
\end{tikzpicture}

\caption{A graph $G$ with $2 = g(G) < 3 = eg(G) < 4= seg(G) < 5 = meg(G).$ 
Note that, a minimum geodetic set of $G$ is $\{v_3,v_5\}$, a minimum edge-geodetic set of $G$  is $\{v_1, v_2, v_4\}$, a minimum strong edge-geodetic set of $G$ is 
$\{v_1,v_2,v_3,v_4\}$ (the assigned  shortest paths between a pair of adjacent vertices 
is the edge between them, and between a pair of non-adjacent vertices is the $2$-path through $v_5$)
and a minimum MEG-set of $G$ is $\{v_1,v_2,v_3,v_4,v_5\}$.}
\label{fig example reln}
\end{figure}
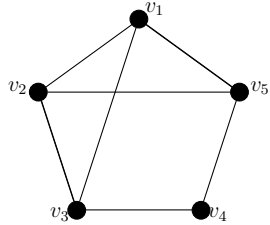

 \begin{theorem}\label{thm Gabcd}
    For any positive integers $4 \leq a \leq b \leq c \leq d$, except for $d = c + 1$ when $b = a+1$ and for $d = c, d = c+2$ when $b \neq a+1$, there exists a connected graph $G_{a,b,c,d}$ with $g(G_{a,b,c,d}) = a, eg(G_{a,b,c,d}) = b, seg(G_{a,b,c,d}) = c$, and $meg(G_{a,b,c,d}) = d$. 
\end{theorem}

\begin{proof}
We begin the proof by describing the construction of $G_{a,b,c,d}$.

\medskip

\noindent \textit{Construction of $G_{a,b,c,d}$:} 
In the first phase of the construction, we start with a $K_{2, 2 + b-a}$, where the partite set of size two has the vertices $x_1$ and $y$, 
and the partite set 
of size $(2+b-a)$ has the vertices $z_1,z_2$ and $w_1, w_2, \ldots, w_{b-a}$. 
Moreover, we add some edges 
in such a way that the set 
$$W = \{w_1, w_2, \ldots, w_{b-a}\}$$ 
becomes a clique. 
We also add the edge 
$z_2w_1$ only if $b=a+1$. 

In the second phase of the construction, 
we add $(c-b+1)$ parallel edges between the vertices 
$z_1$ and $z_2$. After that we subdivide (once) each of the above-mentioned parallel edges and name the degree two vertices created due to the subdivisions as $v_1, v_2, \ldots, v_{c-b+1}$. The set of vertices created by the subdivisions is given by 
$$V=\{v_1, v_2, \ldots, v_{c-b+1}\}.$$

 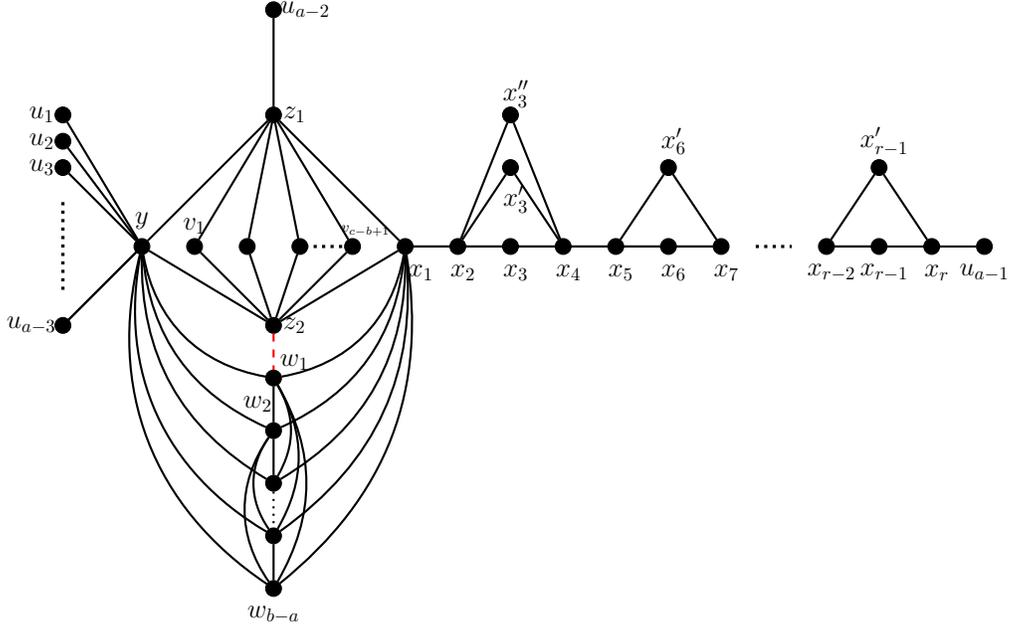
\begin{figure}
\centering
     \begin{tikzpicture}[scale=0.7, transform shape]
     %for node name
     \node at (-1.4,1) {\large$u_3$};
  \node at (-1.4,1.5) {\large$u_2$};
  \node at (-1.4,2) {\large$u_1$};
  \node at (0.5,0) {\large$y$};
  \node at (1.5,-0.1) {\large$v_1$};
  \draw[dotted, very thick] (3.6,-0.5)--(4.5,-0.5); 
    \node at (4.75,-0.2) {\scriptsize $ v_{c-b+1}$};
     \node at (5.8,-1) {\large$x_1$};
      \node at (6.6,-1) {\large$x_2$};
       \node at (7.6,-1) {\large$x_3$};
        \node at (8.6,-1) {\large$x_4$};
         \node at (9.6,-1) {\large$x_5$};
          \node at (10.6,-1) {\large$x_6$};
          \node at (11.6,-1) {\large$x_7$};
        \node at (7.6,0.4) {\large$x'_3$};
        \node at (7.6,2.4) {\large$x''_3$};
         \node at (10.6,1.5) {\large$x'_6$};
          \node at (14.6,1.5) {\large$x'_{r-1}$};
           \node at (14.6,-1) {\large$x_{r-1}$};
            \node at (13.6,-1) {\large$x_{r-2}$};
           \node at (15.6,-1) {\large$x_r$};
  \node at (-1.6,-2) {\large$u_{a-3}$};
 \node at (3.4,2) {\large$z_1$}; 
  \node at (3.4,-2) {\large$z_2$}; 
  \node at (3.6,4) {\large$u_{a-2}$}; 
  \node at (16.5,-1) {\large$u_{a-1}$}; 
  \node at (3.4,-2.7) {\large$w_1$}; 
  \node at (2.7,-3.5) {\large$w_2$}; 
  \node at (3,-7.5) {\large$w_{b-a}$}; 
  %for u_{a-1} 
 \node[circle, draw,inner sep=0.1mm,minimum size=3mm, fill] (u_3) at (-1,1) {};
   \node[circle, draw,inner sep=0.1mm,minimum size=3mm, fill] (u_2) at (-1,1.5) {};
  \node[circle, draw,inner sep=0.1mm,minimum size=3mm, fill] (u_1) at (-1,2) {};
  \node[circle, draw,inner sep=0.1mm,minimum size=3mm,white] (u_4) at (-1,0.5) {};
   \node[circle, draw,inner sep=0.1mm,minimum size=3mm,white] (u_5) at (-1,-1.5) {};
  \node[circle, draw,inner sep=0.1mm,minimum size=3mm, fill] (u_{a-1}) at (-1,-2) {};
 
  %for horizontal
  \node[circle, draw,inner sep=0.1mm,minimum size=3mm, fill] (a) at (0.5,-0.5) {};
  \node[circle, draw,inner sep=0.1mm,minimum size=3mm, fill] (a_1) at (1.5,-0.5) {};
  \node[circle, draw,inner sep=0.1mm,minimum size=3mm, fill] (a_2) at (2.5,-0.5) {};
  \node[circle, draw,inner sep=0.1mm,minimum size=3mm, fill] (a_3) at (3.5,-0.5) {};
  \node[circle, draw,inner sep=0.1mm,minimum size=3mm, fill] (a_4) at (4.5,-0.5) {};
  \node[circle, draw,inner sep=0.1mm,minimum size=3mm, fill] (a_5) at (5.5,-0.5) {};
  \node[circle, draw,inner sep=0.1mm,minimum size=3mm, fill] (a_6) at (6.5,-0.5) {};
  \node[circle, draw,inner sep=0.1mm,minimum size=3mm, fill] (a_7) at (7.5,-0.5) {};
  \node[circle, draw,inner sep=0.1mm,minimum size=3mm, fill] (a_8) at (8.5,-0.5) {};
  \node[circle, draw,inner sep=0.1mm,minimum size=3mm, fill] (a_9) at (9.5,-0.5) {};
  \node[circle, draw,inner sep=0.1mm,minimum size=3mm, fill] (a_10) at (10.5,-0.5) {}; 
  \node[circle, draw,inner sep=0.1mm,minimum size=3mm, fill] (a_11) at (11.5,-0.5) {}; 
  \node[circle, draw,inner sep=0.1mm,minimum size=3mm, fill] (a_14) at (13.5,-0.5) {}; 
  \node[circle, draw,inner sep=0.1mm,minimum size=3mm, fill] (a_15) at (14.5,-0.5) {}; 
  %hide vertices on horizonal
  \node[circle, draw,inner sep=0.1mm,minimum size=3mm, white] (a_12) at (12,-0.5) {}; 
  \node[circle, draw,inner sep=0.1mm,minimum size=3mm, white] (a_13) at (13,-0.5) {};

  \draw[dotted, very thick] (u_4)--(u_5); 
 \draw[thick] (u_3)--(a); 
 \draw[thick] (a)--(u_{a-1}); 
 \draw[thick] (a_5)--(a_11); 
 \draw[dotted,very thick] (a_12)--(a_13); 
 \draw[thick] (a)--(u_1); 
 \draw[thick] (u_2)--(a); 
 \draw[thick] (a)--(u_{a-1}); 
 \draw[thick] (a_5)--(a_11); 
 \draw[thick] (a_14)--(a_15); 
%\path[->] (5,0.8) edge [bend left=30,looseness=1]  (4.6,-0.2);
  %diamondgraph
  \node[circle, draw,inner sep=0.1mm,minimum size=3mm, fill] (z_1) at (3,2) {}; 
  \node[circle, draw,inner sep=0.1mm,minimum size=3mm, fill] (z_2) at (3,-2) {}; 
  \node[circle, draw,inner sep=0.1mm,minimum size=3mm, fill] (z_3) at (3,4) {}; 
  %small diamond graph
  \node[circle, draw,inner sep=0.1mm,minimum size=3mm, fill] (z_4) at (7.5,2) {};
  \node[circle, draw,inner sep=0.1mm,minimum size=3mm, fill] (z_5) at (7.5,1) {}; 
  \node[circle, draw,inner sep=0.1mm,minimum size=3mm, fill] (z_7) at (10.5,1) {}; 

  \node[circle, draw,inner sep=0.1mm,minimum size=3mm, fill] (a_16) at (15.5,-0.5) {};
  \node[circle, draw,inner sep=0.1mm,minimum size=3mm, fill] (a_17) at (16.5,-0.5) {};
  \node[circle, draw,inner sep=0.1mm,minimum size=3mm, fill] (z_8) at (14.5,1) {};
 \draw[thick] (a_14)--(z_8);
 \draw[thick] (a_16)--(z_8);
 \draw[thick] (a_14)--(a_17);
 \draw[thick] (z_1)--(a);
 \draw[thick] (z_1)--(a_4);
 \draw[thick] (z_2)--(a);
 \draw[thick] (z_2)--(a_5);
 \draw[thick] (z_1)--(a_5);
 \draw[thick] (z_2)--(a_4);
 \draw[thick] (z_1)--(a_1);
 \draw[thick] (z_1)--(a_2);
 \draw[thick] (z_1)--(a_3);
 \draw[thick] (z_2)--(a_2);
 \draw[thick] (z_2)--(a_1);
 \draw[thick] (z_2)--(a_3);
 \draw[thick] (z_1)--(z_3);
 \draw[thick] (a_6)--(z_4);
 \draw[thick] (a_6)--(z_5);
 \draw[thick] (a_8)--(z_5);
 \draw[thick] (a_8)--(z_4);
 \draw[thick] (a_9)--(z_7);
 \draw[thick] (a_11)--(z_7);
 
 %clique graph
 \node[circle, draw,inner sep=0.1mm,minimum size=3mm, fill] (c_1) at (3,-3) {}; 
  \node[circle, draw,inner sep=0.1mm,minimum size=3mm, fill] (c_2) at (3,-4) {}; 
  \node[circle, draw,inner sep=0.1mm,minimum size=3mm, fill] (c_3) at (3,-5) {}; 
   \node[circle, draw,inner sep=0.1mm,minimum size=3mm, fill] (c_4) at (3,-6) {}; 
  \node[circle, draw,inner sep=0.1mm,minimum size=3mm, fill] (c_5) at (3,-7) {}; 
 \draw[thick] (c_1)--(c_2);
 \draw[thick] (c_2)--(c_3);
  \draw[thick, dashed,red] (z_2)--(c_1);
   \draw[dotted,thick] (c_3)--(c_4);
    \draw[thick] (c_4)--(c_5);
 \path[thick] (a) edge [bend right=35,looseness=1] (c_1);
 \path[thick] (a_5) edge [bend left=35,looseness=1] (c_1);
 \path[thick] (a) edge [bend right=30,looseness=1] (c_2);
 \path[thick] (a_5) edge [bend left=30,looseness=1] (c_2);
 \path[thick] (a) edge [bend right=30,looseness=1] (c_3);
 \path[thick] (a_5) edge [bend left=30,looseness=1] (c_3);
\path[thick] (c_1) edge [bend left=30,looseness=1] (c_3);

 \path[thick] (a_5) edge [bend left=30,looseness=1] (c_4);
\path[thick] (c_1) edge [bend left=30,looseness=1] (c_4);
 \path[thick] (a_5) edge [bend left=30,looseness=1] (c_5);
\path[thick] (c_1) edge [bend left=30,looseness=1] (c_5);
\path[thick] (a) edge [bend right=35,looseness=1] (c_4);
\path[thick] (a) edge [bend right=35,looseness=1] (c_5);

\path[thick] (c_2) edge [bend right=35,looseness=1] (c_4);
\path[thick] (c_2) edge [bend right=35,looseness=1] (c_5);
\end{tikzpicture}
\caption{The structure of $G_{a,b,c,d}.$ The red dashed edge between $z_2$ and $w_1$ represents an edge/non-edge depending on the values of $a$ and $b$.}
\label{fig Gabcd first}
\end{figure}

In the third phase of the construction, we add $(a-3)$ pendent neighbors $u_1, u_2, \ldots, u_{a-3}$ to $y$, and one pendent neighbor $u_{a-2}$ to $z_1$. 
Moreover, we attach a long path $x_1x_2 \ldots x_r u_{a-1}$
with the vertex $x_1$, where $u_{a-1}$ is a pendent vertex and  $r = 3\left\lfloor \frac{d-c}{2} \right\rfloor +1$. 
Next we will add a open twin $x'_{3i}$ to the vertices 
of the form $x_{3i}$ for all $i \in \left\{1, 2, \cdots, \left\lfloor \frac{d-c}{2} \right\rfloor\right\}$. Additionally, if $(d-c)$ is odd, then we will add another twin $x''_{3}$ to the vertex $x_3$. 
For convenience, let
$$U = \{u_1, u_2, \ldots, u_{a-1}\}$$ 
 denote the set of all pendent vertices and 
$$X = \left\{x_{3i}, x'_{3i} ~|~ i = 1, 2, \ldots,\left\lfloor \frac{d-c}{2} \right\rfloor \right\} \cup \{x''_3 \}$$ 
denote the set of all degree $2$ vertices on the paths connecting $x_1$ and $u_{a-1}$.

Note that, $x''_3$ exists in $X$ if and only if $(d-c)$ is odd. Similarly, when $b \neq a+1$, we maintain the construction same as above except for a small change where we take $r = 3\lfloor\frac{d-c-1}{2}\rfloor + 1$, and $x''_3$ exists in $X$ if and only if $(d-c)$ is even. This completes the description of the construction of the graph $G_{a,b,c,d}$ (see Fig.~\ref{fig Gabcd first} for a pictorial reference). 

As $U$ is the set of all pendents, we know that it will be 
part of any geodetic set, edge-geodetic set, strong edge-geodetic set, and monitoring edge-geodetic set by Lemma~\ref{lemma:simpl}. However, the vertices of $U$ cannot cover the vertices of $V$ using shortest paths between the vertices of $U$. Therefore, we need at least one more vertex to form a geodetic set. As $U \cup \{z_2\}$ is a geodetic set of $G$, we can infer that 
$$g(G_{a,b,c,d}) = |U| + |\{z_2\}| = (a-1)+1=a.$$

Next, observe that the vertices of $U$ are not able to 
cover any edge of the clique $W$ using shortest paths between the vertices of $U$. Moreover,  the only way to monitor those edges is by taking $W$ in our edge-geodetic set. Similarly none of the edges of the form $v_iz_j$ are monitored with only vertices from $U$ and $W$. On the other hand, $U \cup W \cup \{z_2\}$ is an edge-geodetic set when $b \geq a+2$. Note that when $b=a+1$, $z_2$ and $w_1$ are adjacent. In this case too, we need the vertices from the set $U$ and the vertex $z_2$ to monitor all the edges except $z_2w_1$. Thus, we take $W = \{w_1\}$ in our edge-geodetic set.
 Therefore, 
$$eg(G_{a,b,c,d}) = |U| + |W| + |\{z_2\}| = (a-1)+(b-a)+1=b.$$

We know that the vertices of $U$ are in any strong edge-geodetic set. Moreover, note that, the vertices of $W$ must be in any strong edge-geodetic set to cover the edges of the clique $W$ if $b\neq a+1$, or to monitor $w_1z_2$ if $b=a+1$. 
Now, let us see how we can cover the edges of the $(c-b+1)$ $2$-paths between $z_1, z_2$, having the vertices of $V$ as their internal vertex. We can assign one shortest path involving one of $v_i'$s between $u_{a-2}$ and $z_2$, but for the remaining edges, we must take $c-b$ vertices in our strong edge-geodetic set to cover those incident edges. 
% First of all, if we do not take $z_2$ in our strong edge-geodetic set, we have to take all vertices of $V$. Second of all, if we take $z_2$ in our 
% strong edge-geodetic set, then we have to take either (at least) all but one vertices of $V$, or all but two vertices of $V$ along with $z_1$ in the strong edge-geodetic set. That means, we need to take at least $(c-b+1)$ additional vertices in the strong edge-geodetic set.
Moreover, 
the set $U \cup W \cup (V \setminus \{v_1\}) \cup \{z_2\}$ is indeed a strong edge-geodetic set, as every other edge can be monitored by assigning a shortest path between $u_1$ to $u_{a-1}$ using all $x_{3i}$, a shortest path from $u_{a-2}$ to $u_{a-1}$ using all $x'_{3i}$ and a shortest path from $u_{a-2}$ to $u_{a-1}$ using $x''_{3i}$ (if it exists). Thus, 
$$seg(G_{a,b,c,d})=|U|+|W|+|V-1|+ |\{z_2\}|=(a-1)+(b-a)+(c-b)+1=c.$$

Finally, for any MEG-set, we have to take the vertices of 
$U$ (as they are pendents), the vertices of $W$ (to monitor the edges of the clique $W$), the vertices of $X$ (as they are twins, see Lemma~\ref{lemma:twins}). However, even with these vertices, we cannot monitor the edges of the 
$(c-b+1)$ $2$-paths between $z_1, z_2$, having the vertices of $V$ as their internal vertex. To do so, we have to take all vertices of $V$ in our MEG-set. These vertices also help monitor the edge $z_2w_1$. Hence, the set 
$U \cup W \cup V \cup X$ is an MEG-set when $b = a+1$,  Hence, 
$$meg(G_{a,b,c,d})=|U|+|W|+|V|+|X|=(a-1)+1+(c-b+1)+2\left\lfloor \frac{d-c}{2} \right\rfloor + \epsilon = d,$$
where $\epsilon = 0$ (resp., $1$) if $(d-c)$ is (even (resp., odd), and $d \neq c+1$. 
Observe that, when we have $b \neq a+1$, that is, when $b=a$ or when $b \geq a+2$, we do not have the edge $z_2w_1$, hence, we have to take $z_2$ in our MEG-set to monitor the edges of the type $v_iz_j$. In this case, the set $U \cup V \cup W \cup X \cup \{z_2\}$ is an MEG-set. Hence, 
\begin{align*}
    meg(G_{a,b,c,d})&=|U|+|W|+|V|+|X|+|\{z_2\}|\\ &=(a-1)+(b-a)+(c-b+1)+2\left\lfloor \frac{d-c-1}{2} \right\rfloor + \epsilon + 1 \\ &= d,
\end{align*}

where $\epsilon = 0$ (resp., $1$) if $(d-c)$ is (odd (resp., even), and $d \neq c, c+2$. 

This completes the proof.
\end{proof}

\subsection{Relation with distance-edge monitoring sets}
Observe that the concept of MEG-set is closely related to that of distance-edge monitoring set. Notice that, MEG-sets are particular types of 
distance-edge monitoring sets (see the definition of distance-edge monitoring set in Subsection~\ref{subsec:defs}), and hence we have
$dem(G) \leq meg(G)$. However, we show that this inequality is strict.  

\begin{lemma}
    For any graph $G$ having at least one edge, we have $dem(G) < meg(G)$.
\end{lemma}

\begin{proof}
    Let $M$ be a minimum MEG-set of $G$. 
    We claim that $M \setminus \{x\}$ for any $x \in M$ is a distance-edge monitoring set of $G$. This is correct as any edge of $G$ is monitored by two vertices of $M$, and at least one of them must belong to $M \setminus \{x\}$. Thus $dem(G) \leq |M| - 1 < meg(G).$ 
\end{proof}

In view of the above lemma, given any two positive integers $p < q$, we prove the existence of a graph $G_{p,q}$ satisfying $dem(G_{p,q}) = p$ and $meg(G_{p,q}) = q.$ 
Note that it is possible to have $dem(G)=1$ for a graph $G$, (for example paths~\cite{foucaud2022monitoring}). Whereas $meg(G) \geq 2$ for all graphs $G$ having at least one edge.

\begin{theorem}
    For any positive integers $1 \leq p < q$, there exists a connected graph $G_{p,q}$ with $dem(G_{p,q}) = p$ and $meg(G_{p,q}) = q.$
\end{theorem}

\begin{proof}
    We provide two separate constructions for the proof. One  for the case when $p=1$,  and the other for the case when $p \geq 2$. 

    \medskip 

      \noindent \textit{Case 1: When $1 = p < q$:}     
    Consider the star graph $K_{1,q}$ with  partite sets
    $\{y\}$ and $\{u_1, u_2, \cdots u_q \}$. 
    Notice that $\{y\}$ is a minimum distance-edge monitoring set of $G$~\cite{foucaud2022monitoring}. Thus $dem(G_{p,q}) = dem(K_{1,q}) = 1=p$. 
    On the other hand, the set of all leaves, that is, $\{u_1, u_2, \cdots u_q \}$  is a 
    minimal MEG-set~\cite{foucaud2023monitoring}. Hence $meg(G_{p,q}) = meg(K_{1,q}) = q$. 

    \medskip

    \noindent  \textit{Case 2: When $2 \leq p < q$:}
   Let $u_1,u_2, \cdots u_{p+1}$ be the vertices of a complete graph $K_{p+1}$. We add $(q-p)$ pendent vertices in the neighborhood of $u_1$. The so-obtained graph is $G_{p,q}$. 
    We know that a minimum distance-edge monitoring set of 
    $G_{p,q}$ is 
    $\{ u_2, u_3, \cdots, u_{p+1}\}$~\cite{foucaud2022monitoring}, and a minimum MEG-set of this graph is $\{v_1, v_2, \cdots v_{q-p}, u_2, u_3, \cdots u_{p+1}  \}$~\cite{foucaud2023monitoring}. 
    Hence $dem(G_{p,q}) = p$ and $meg(G_{p,q}) = q$. 
\end{proof}

\section{Conditions for a vertex being in all or no optimal MEG-sets}\label{sec megfull}
In their introductory paper on monitoring edge-geodetic sets, Foucaud, Krishna and Ramasubramony Sulochana~\cite{foucaud2023monitoring}   asked to characterize 
the graphs $G$ having $meg(G) = |V(G)|$. 
We provide a definitive answer to their question, and to this end, we give a necessary and sufficient condition for a vertex to be in any MEG-set of a graph. An \textit{induced $2$-path} of a graph $G$ is an ordered set of three vertices $u, v, x$ such that $u, x$ are adjacent to $v$ while  $u$ is not adjacent to $x$.

\begin{theorem}\label{thm meg full}
Let $G$ be a graph. A vertex $v\in V(G)$ is in every MEG-set of $G$ if and only if there exists $u \in N(v)$ such that for any vertex $x\in N(v)$, any induced $2$-path $uvx$ is part of a $4$-cycle.
\end{theorem}

\begin{proof}
    For the necessary condition, let us assume that a vertex $v \in V(G)$ is in every MEG-set of $G$. We have to prove that there exists $u\in N(v)$, such that any induced $2$-path $uvx$ is part of a $4$-cycle. We prove it by contradiction.
    
    Suppose for every $u\in N(v)$, there exists an induced $2$-path $uvx$ such that $uvx$ is not part of a $4$-cycle. As $uvx$ is an induced $2$-path, observe that $u$ and $x$ are not adjacent, also as $uvx$ is not part of a $4$-cycle, we get, $d(u,x)=2$ and the only shortest path between $u$ and $x$ is via $v$. This implies that, if we take $u$ and $x$ in our MEG-set $S$, then $xv$ and $uv$ are monitored. 
    Hence, all the neighbors of $v$ can monitor all the edges incident to $v$. 
    Therefore, in particular, 
    $V(G) \setminus \{v\}$ is an MEG-set of $G$. This is a contradiction to the fact that $v$ is in every 
    MEG-set of $G$. 
    Thus, the necessary condition for a vertex $v$ to be part 
    of every MEG-set of $G$ is proved according to the statement.

    For the sufficient condition, let us assume that for some vertex $v$ of $G$, there exists $u \in N(v)$ such that any induced $2$-path $uvx$ is part of a $4$-cycle. We need to prove that $v$ is in every MEG-set. Thus, it is enough to show that 
    $S = V(G) \setminus \{v\}$ is not an MEG-set of $G$. Therefore, we would like to find an edge which is not monitored by the vertices of $S$. 
    
    We first observe that if there does not exist any induced $2$-path of the form $uvx$, then $v$ must be a simplicial vertex, and thus we know that $v$ belongs to every MEG-set of $G$ by Lemma~\ref{lemma:simpl}. On the other hand, if there exists an induced $2$-path of the form $uvx$, then according to our assumption, there exists a $4$-cycle of the form 
    $uvxwu$. Assume there exist $a, b$ two vertices of $G$ monitoring the edge $uv$ (where $a$ and $u$ could be the same vertex). We consider a shortest path from $a$ to $b$ that we denote $P$, where $P=a\dots uvx\dots b$ (where $x$ and $b$ could be the same vertex). Note that $P'=a\dots uwx\dots b$ is another shortest path from $a$ to $b$. Hence $a, b$ do not monitor the edge $uv$, which implies that $v$ has to be part of every MEG-set. This concludes the proof. 
    \end{proof}

An immediate corollary characterizes all graphs $G$ with $meg(G) = |V(G)|$. This answers an open question asked in~\cite{foucaud2023monitoring}.      

\begin{corollary}\label{cor:all}
    Let $G$ be a graph of order $n$. Then, $meg(G)=n$ if and only if for every $v\in V(G)$, there exists $u\in N(v)$ such that any induced $2$-path $uvx$ is part of a $4$-cycle.
\end{corollary}

\begin{proof}
    The proof directly follows from Theorem~\ref{thm meg full}, since if there exists $u\in V(G)$ that does not fulfill the condition, then $V(G) \setminus \{u\}$ would be an MEG-set.
\end{proof}

A \textit{universal} vertex $u$ of a graph $G$ is a vertex that is adjacent to every vertex except 
itself in $G$.

\begin{corollary}\label{cor:uvertex}
    Let $G$ be a graph with a universal vertex $u$. Then any vertex $v \neq u$ is in every MEG-set of $G$. In particular, if $G$ has $n$ vertices, then $meg(G)\geq n-1$.
\end{corollary}
\begin{proof}
    Since $u$ is a universal vertex, for any $v, x\in V(G)$, $uvx$ is never an induced $2$-path as $u, x$ are adjacent. Therefore, every vertex $v \neq u$ satisfies the necessary condition for being part of every MEG-set of $G$ according to Theorem~\ref{thm meg full}. 
\end{proof}

The \textit{girth} of a graph is the length of its smallest cycle.

\begin{corollary}\label{cor girth 5}
    Let $G \neq K_1, K_2$ be a 
    connected graph with girth at least $5$. 
    If $G$ has $n$ vertices, 
    then $meg(G) \leq n-1$.
\end{corollary}

\begin{proof}
Notice that, if a vertex $v$ of $G$ satisfies the condition of Theorem \ref{thm meg full}, then either $v$ has to be part of a $3$-cycle, or a $4$-cycle, or $v$ is a pendent vertex. As $G$ has girth at least $5$, it cannot contain any $3$-cycle or $4$-cycle. 
Thus, if $G$ has $meg(G) = n$, then all vertices of $G$ are pendent vertices. This is only possible when 
$G = K_2$, which is not possible due to our assumption.  Thus, not every vertex of $G$ can satisfy the condition of Theorem~\ref{thm meg full}. Hence, $meg(G) \leq n-1$.
\end{proof}
\begin{remark}Note that this bound is optimal for star as the girth of a star is infinite.
\end{remark}

\section{MEG-extremal graphs}\label{sec:ext}
An \textit{MEG-extremal graph} $G$ is a graph having $meg(G)=|V(G)|$. 
In~\cite{foucaud2023monitoringfull,foucaud2023monitoring}, the following families of \emph{MEG-extremal} graphs were presented: complete graphs, complete multipartite graphs (except stars), and hypercubes. In~\cite{haslegrave2023monitoring}, the author showed that for any MEG-extremal graph $G$ and any graph $H$, both the Cartesian product and the strong product of $G$ and $H$ are MEG-extremal graphs. 

In this section we extend the same lines of work, primarily using the tools built in Section~\ref{sec megfull}. 
We organize the section in two subsections. In the first subsection we provide  complete characterization of $meg(G)$ when $G$ belongs to the families of cographs, block graphs, well-partitioned chordal graphs, split graphs, and proper interval graphs. Using our characterization we then conclude that the $2$-connected graphs in the above mentioned families are MEG-extremal. In the second subsection we present a short proof of a result from~\cite{haslegrave2023monitoring} where we show that the Cartesian product and the strong product of 
an MEG-extremal graph with any other graph is also MEG-extremal. Moreover, we  show that the tensor product of two MEG-extremal graphs is always MEG-extremal.

\subsection{Graph families}
We start by recalling a useful lemma from~\cite{foucaud2023monitoring} on cut-vertices. 

\begin{lemma}{(Lemma 2.3 of \cite{foucaud2023monitoring})}\label{lemma:curv}
Let $G$ be a graph, and $u$ be a cut-vertex of $G$. Then $u$ is never part of any minimal MEG-set of $G$.
\end{lemma}

\begin{remark}\label{rem no deg 1 in minMEG}
    Due to Lemma~\ref{lemma:curv}, for a connected graph $G \neq K_2$, every vertex of degree $1$ must be part of any MEG-set, and its neighbor cannot be
    a part of any minimum MEG-set. Therefore, if $meg(G) = |V(G)|$, then $G$ must have minimum degree at least $2$. 
\end{remark}

A \textit{cograph} $G$ is a graph which does not contain any induced path on $4$ vertices, that is, an induced $P_4$. A \textit{complete join} of two graphs $G_1$ and $G_2$ is the graph $G$ obtained by adding an edge between each vertex of $G_1$
and each vertex of $G_2$.

\begin{theorem}\label{thm:extcographs}
Let $G$ be a connected cograph on $n$ vertices. Then either $G$ has a cut-vertex and $meg(G) =n-1$, or else $meg(G)=n$.
\end{theorem}

\begin{proof}
First assume that $G$ has a cut vertex $u$. Note that, $G - u$ has at least two components. 
Observe that $u$ must be a universal vertex as otherwise it is possible to find an induced $P_4$. 
Thus, $u$ must be the only cut-vertex of $G$. Hence, by Corollary~\ref{cor:uvertex} and Lemma~\ref{lemma:curv} $meg(G) = n-1$.

A $2$-connected cograph is of diameter at most $2$. 
Take any vertex  $v$ and any of its neighbors $u$. 
If there is an induced $2$-path of the form $uvx$ in $G$, then by $2$-connectedness there exists an induced path $P$ connecting $u$ and $x$ which is internally vertex disjoint with $uvx$. However, $P$ must have less than $4$ vertices, and thus, either $P$ is an edge, or an induced $2$-path. Therefore, $v$ must be part of every MEG-set of $G$ according to 
Theorem~\ref{thm meg full}. 
\end{proof}

A \textit{block graph} $G$ is a graph all of whose $2$-connected induced subgraphs are cliques.

\begin{theorem}\label{thm:extblockgraphs}
Let $G$ be a block graph with $k$ cut-vertices. If $G$ has $n$ vertices, then $meg(G)=n-k$.
\end{theorem}

\begin{proof}
In a block graph, a vertex is either simplicial or a cut-vertex. Hence, the result follows from Lemma~\ref{lemma:simpl} and Lemma~\ref{lemma:curv}.
\end{proof}

A graph is \textit{chordal} if it does not contain any induced cycle of length at least $4$. A graph $G$ is a \textit{well-partitioned chordal graph}~\cite{AHN2022112985} if its vertex set can be partitioned into cliques $C_1, C_2, \ldots, C_{\ell}$ that satisfies the following properties: 
\begin{enumerate}[(i)]
    \item The cliques $C_1, C_2, \ldots, C_{\ell}$ are called bags. 

    \item Two bags $C_i$ and $C_j$ are non-adjacent if there is no edge with one end point in $C_i$ and the other in $C_j$. 

    \item  Two bags $C_i$ and $C_j$ are adjacent
    if the edges between $C_i$ and $C_j$ induce a complete bipartite graph. The complete bipartite graph does not need to span the two bags. 

    \item Two bags are either adjacent or non-adjacent and the graph obtained by considering bags as vertices with the above-mentioned adjacency rule is a forest. 
\end{enumerate}

\begin{theorem}\label{thm:extwpchordalgraphs}
Let $G$ be a well-partitioned chordal graph with $k$ cut-vertices. If $G$ has $n$ vertices, then 
$meg(G)=n-k$.
\end{theorem}

\begin{proof}
Let $v$ be a vertex of $G$ which is neither a cut-vertex nor a simplicial vertex. It is enough to show that there exists a neighbor $u$ of $v$ such that any induced $2$-path of the form $uvx$ is part of a $4$-cycle due to Lemmas~\ref{lemma:simpl} and~\ref{lemma:curv}. Assume that $v$ belongs to the bag $C_i$.

Note that, as $v$ is not a simplicial vertex, then $v$ is adjacent to a vertex in a different bag $C_k$. We also assumed $v$ was not a cut-vertex, hence there exists an edge incident to a vertex of $C_k$ and a vertex of $C_i$ distinct from $v$, hence $C_i$ must contain at least one more vertex in it, say $u$. 
Suppose $uvx$ is an induced $2$-path. As $u$ and $x$ are non-adjacent, $x$ must belong to a different bag $C_j$. Moreover, as $v$ is not a cut vertex, there must be another vertex $w$ in $C_i$ which has an edge with a vertex of $C_j$. Thus, due to the definition of a well-partitioned chordal graph, $w$ must be adjacent to both $u$ and $x$. That is, $uvxwu$ is a $4$-cycle. 
Hence, $v$ must be part of all MEG-sets of $G$ due to Theorem~\ref{thm meg full}. 
\end{proof}

A \textit{split graph} $G$ is graph whose vertices can be partitioned into a clique and an independent set.

\begin{corollary}
    Let $G$ be a split graph with no component equal to $K_1$ or $K_2$ with $k$ vertices having a pendent neighbor. If $G$ has $n$ vertices, then 
$meg(G)=n-k$.
\end{corollary}

\begin{proof}
    Let $G$ be a split graph on $n$ vertices with $k$ vertices having a pendent neighbor. Suppose that 
    $G$ can be partitioned into a clique $C$ and an independent set $I$. Now consider $C$ as one bag, and each vertex of $I$ as a bag. This shows that $G$ is a well-partitioned graph as well. The proof follows from Theorem~\ref{thm:extwpchordalgraphs}  observing that the only cut-vertices in a split graph are those having a pendent neighbor. 
\end{proof}

An \textit{interval graph} $G$ is graph whose vertices correspond to intervals, and two vertices are adjacent if and only if their corresponding intervals intersect. Moreover, if it is possible to assign an interval to each vertex of $G$ in such a way that none of the intervals contains another interval, then $G$ is a \textit{proper interval graph}.

\begin{theorem}\label{thm:extpropintervalgraphs}
    Let $G$ be a proper interval graph with $k$ 
    cut-vertices. If $G$ has $n$ vertices, 
    then $meg(G)=n-k$.
\end{theorem}

\begin{proof}
   Let $G$ be a proper interval graph on $n$ vertices with $k$ cut-vertices. For any vertex $z$ in $G$, let us denote its corresponding interval by 
   $[l_z, r_v]$ where $l_z$ denotes the left endpoint of the interval and $r_z$ denotes the right endpoint of the interval. As $G$ is a proper interval graph, it is possible to provide a total ordering $\prec$ on the vertices of $G$ by defining 
   $z \prec z'$ if $l_z < l_{z'}$. Observe that, as $G$ is a proper interval graph, $l_z < l_{z'}$ if and only if $r_z < r_{z'}$.

   Let $v$ be a vertex of $G$ which is neither a cut-vertex nor a simplicial vertex. It is enough to show that there exists a neighbor $u$ of $v$ such that any induced $2$-path of the form $uvx$ is part of a $4$-cycle due to Lemmas~\ref{lemma:simpl} and~\ref{lemma:curv}.
    
  As $v$ is not a simplicial vertex, $v$ has neighbors 
 which intersects $l_v$, 
 and neighbors which intersects $r_v$. 
 Let $L_v$ denote the set of all neighbors $y$ of $v$ 
 such that $y \prec v$. Let $u$ be the maximum element of $L_v$ with respect to $\prec$. 
We will show that, any induced $2$-path of the form $uvx$ must be part of a $4$-cycle. 

Let $uvx$ be an induced $2$-path. Notice that we must have $u \prec v \prec x$. As $v$ is not a cut-vertex, there exists an induced path $P = uw_iw_2\ldots w_{\ell}x$ connecting $u$ and $x$ which does not contain the vertex $v$. Notice that as $P$ is an induced path, we must have 
$$u \prec w_1 \prec w_2 \cdots \prec w_{\ell} \prec x.$$ 
That means, every $w_i$ is a neighbor of $v$. Moreover, due to the maximality of $u$, we must have $v \prec w_1$. 
That means $w_1$ intersects $r_v$, and hence intersects $[l_x,r_x]$. Hence, $uvxw_1u$ is a $4$-cycle. 
\end{proof}

The results proved in this section together imply the following theorem:

\begin{theorem}\label{th:2-connected is extremal}
    If $G$ is a $2$-connected cograph, block-graph, well-partitioned chordal graph, split graph, or  proper interval graph, then $G$ is MEG-extremal.
\end{theorem}

\begin{remark}
    As split graphs are a subclass of well-partitioned chordal graphs, it was not necessary to mention them separately in Theorem~\ref{th:2-connected is extremal}. However, we still do so since split graphs are a well-known class of graphs.
\end{remark}

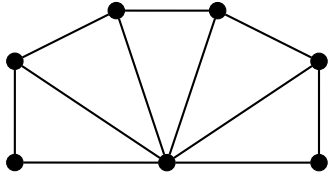
\begin{figure}
    \centering
    \begin{tikzpicture}[
      node distance = 8mm and 8mm,
state/.append style = {circle, draw, fill, 
                       minimum size=0.1em},
every edge/.style = {draw, thick}, transform shape, scale=0.7] 

\begin{scope}[nodes=state, scale=.8]
\node (n1)  {};
\node (n2)  [right=of n1, fill=white, draw=white]   {};
\node (n3)  [right=of n2, fill=white, draw=white]   {};
\node (n4)  [right=of n3]   {};
%%\node (n4)  [below=of n1, fill=white, draw=white]   {};
%
\node (n5)  [above=of n4,fill=white, draw=white]   {};
\node (n6)  [above =of n5]   {};
\node (n7)  [left=of n6, fill=white, draw=white]   {};
\node (n8)  [left=of n7,fill=white, draw=white]   {};
%\node (n8)  [below =of n1, fill=white, draw=white]   {};
%
\node (n9)  [above=of n8]   {};
\node (n10)  [left=of n9, fill=white, draw=white]   {};
\node (n11)  [left=of n10]   {};
\node (n12)  [below=of n11, fill=white, draw=white]   {};
\node (n13)  [left=of n12, fill=white, draw=white]   {};
\node (n14)  [left=of n13]   {};

\node (n15)  [below=of n14, fill=white, draw=white]   {};
\node (n16)  [below=of n15]   {};
%\node (n4)  [below=of n1, fill=white, draw=white]   {};
\end{scope}
\path   (n1) edge  (n4)
(n1) edge  (n6)
(n1) edge  (n9)
(n1) edge  (n11)
(n1) edge  (n14)
(n1) edge  (n16)
(n4) edge  (n6)
(n6) edge  (n9)
 (n9) edge  (n11)
(n11) edge  (n14)
(n14) edge  (n16)
%(n14) edge  (n9)
;

\end{tikzpicture}

    \caption{A $2$-connected interval graph with a non-extremal MEG-set. Note that the set of all vertices except the universal vertex is an MEG-set.}
    \label{fig:chordal-counterexemple}
\end{figure}

\begin{remark}
   A natural question to ask  is whether we can include superclasses of the graph families mentioned in 
   Theorem~\ref{th:2-connected is extremal}. In Figure \ref{fig:chordal-counterexemple} we provide an example of 
   a $2$-connected interval graph which is not MEG-extremal. 
   As interval graphs are also chordal graphs, we cannot hope to extend Theorem~\ref{th:2-connected is extremal} to the family of interval graphs or chordal graphs. 
\end{remark}

\subsection{Graph products}
Let $G$ and  $H$ be two  graphs. Now we are going to define 
three (product) graphs, each on the set of vertices 
$$V(G) \times V(H) = \{(a, b): a \in V(G) \text{ and } b \in V(H) \}.$$

The \textit{strong product} of $G$ and $H$, denoted by $G \boxtimes H$, has the following adjacency rule: two vertices $(a, b)$ and 
$(a', b')$ are adjacent if and only if 
one of the following conditions are satisfied: (i) $a = a'$ and $bb' \in E(H)$, 
(ii) $aa' \in E(G)$ and $b=b'$, 
(iii) $aa' \in E(G)$ and $bb' \in E(H)$.

The \textit{Cartesian product} of $G$ and $H$, denoted by $G ~\Box~ H$, has the following adjacency rule: two vertices $(a, b)$ and 
$(a', b')$ are adjacent if and only if 
either $a = a'$ and $bb' \in E(H)$,
or $aa' \in E(G)$ and $b=b'$.

The \textit{tensor product} of $G$ and $H$, denoted by $G \times H$, has the following adjacency rule: two vertices $(a, b)$ and 
$(a', b')$ are adjacent if and only if 
$aa' \in E(G)$ and $bb' \in E(H)$.

Our next result provides a shorter 
proof of Corollary~2 of~\cite{haslegrave2023monitoring}.

\begin{theorem}[\cite{haslegrave2023monitoring}]
    Let $G$ be an MEG-extremal graph and let $H$ be any graph. Then their 
    Cartesian product $G~\Box~H$ and strong product  
    $G\boxtimes H$ are both MEG-extremal. 
\end{theorem}

\begin{proof}
    We prove this equality for the Cartesian product $G ~\Box~ H$ first. 
    Let $(a, b)$ be a vertex of $G~\Box~ H$. 
    Since $G$ is MEG-extremal, 
    by Theorem \ref{thm meg full}, for the vertex $a \in V(G)$, there exists $u \in V(G)$ such that every induced $2$-path $uax$ is part of a $4$-cycle. We will show that, all induced 
    $2$-path of the form $(u,b)(a,b)(x,y)$ 
    is part of a $4$-cycle in $G ~\Box~ H$. 
    Since $(a,b)$ is adjacent to $(x,y)$, we have the following cases:
    \begin{itemize}
        \item Let $ax \in E(G)$ and $b=y$. 
        First notice that if $(u, b)(a, b)(x, y)$ is an induced $2$-path, $(u, b)$ and $(x, y)$ are not adjacent. Since $b=y$, the vertices $x, u$ are not adjacent in $G$. 
        Thus, $uax$ is an induced $2$-path in $G$. This implies that there exists a $w \in V(G)$ such that 
        $uaxwu$ is a $4$-cycle in $G$. 
        Hence, observe that $(w, b)$
        is adjacent to both 
        $(u, b)$ and $(x, y)$. 
        
        \item Let $a=x$ and $by \in E(H)$. 
        In this case note that the vertex $(u,y)$
        is adjacent to both $(u,b)$ and $(x,y)$. 
    \end{itemize}

     Thus, by Theorem \ref{thm meg full}, $(a,b)$ is in every MEG-set of $G~\Box~ H$, and hence $G ~\Box~ H$ is MEG-extremal.

    For proving that the strong product $G \boxtimes H$ is also MEG-extremal, we may consider the above cases here. But the second case is redundant for strong product of graphs as $(u,b)$ is adjacent to $(a,y)$ and the path is no longer induced. That apart, we need to consider a third case, which is as follows. 
    \begin{itemize}
        \item Let $ax\in E(G)$ and $by\in E(G)$. We know that there exists a vertex  $w\in V(G)$ such that $uaxw$ is a $4$-cycle in $G$. Then observe that 
        $(w,y)$ is adjacent to both $(u,b)$ and $(x,y)$.
    \end{itemize}

    Thus, by Theorem \ref{thm meg full}, $(a,b)$ is in every MEG-set of $G\boxtimes H$, and hence $G \boxtimes H$ is MEG-extremal.
   
\end{proof}

\begin{theorem}
    Let $G, H$ be two MEG-extremal graphs. Then their tensor product $G \times H$ is also 
    MEG-extremal. 
\end{theorem}

\begin{proof}
    Let $(a,b)$ be a vertex of $G \times H$. As $G$ is MEG-extremal, there exists a $u \in V(G)$ such that any induced $2$-path of the form $uax$ is part of a $4$-cycle
    $uaxwu$ in $G$. 
    Similarly, as $H$ is MEG-extremal, 
    there exists a $v \in V(H)$ such that any induced $2$-path of the form $vby$ is part of a $4$-cycle
    $vbyzv$ in $H$.

    Let us consider an induced $2$-path of the form $(u,v)(a,b)(x,y)$ in $G \times H$. Notice that at least one of $uax$ and $vby$ is an induced $2$-path in $G$ or $H$, respectively. If both $uax$ and $vby$ are induced $2$-paths, then due to the observations noted in the last paragraph, there exists $w$ that is adjacent to  both $u,x$ in $G$ and there exists $z$ that is adjacent to both $v,y$ in $H$. Hence, $(w,z)$ is adjacent to both $(u,v)$ and $(x,y)$. Otherwise, without loss of generality, we assume that $uax$ is an induced two path and consider $w$ adjacent to  both $u,x$ in $G$, and observe that $(w,b)$ is adjacent to both $(u, v)$ and $(x, y)$, which concludes the proof.
\end{proof}

\begin{remark}
One can check that the tensor product of an MEG-extremal graph with any graph does not always yield an MEG-extremal graph. For example, the tensor product of the MEG-extremal graph $K_2$ and the path $P_3$ on $3$ vertices (see Figure \ref{fig:tensor}) is not MEG-extremal.
\end{remark}
\begin{figure}
    \centering
    \begin{tikzpicture}[scale=.8, inner sep=.7mm]
        \node[black, circle, fill] (0) at (0, 0) []{};
        \node[black, circle, fill] (1) at (0, 1) []{};
        \node[black, circle, fill] (2) at (0, 2) []{};
        \node[black, circle, fill] (3) at (1, 0) []{};
        \node[black, circle, fill] (4) at (1, 1) []{};
        \node[black, circle, fill] (5) at (1, 2) []{};
        \node[] () at (-1, 0) []{$(u,a)$};
        \node[] () at (-1, 1) []{$(u,b)$};
        \node[] () at (-1, 2) []{$(u,c)$};
        \node[] () at (2, 0) []{$(v,a)$};
        \node[] () at (2, 1) []{$(v,b)$};
        \node[] () at (2, 2) []{$(v,c)$};
        \draw[](0)--(4){};
        \draw[](2)--(4){};
        \draw[](3)--(1){};
        \draw[](5)--(1){};
    \end{tikzpicture}
    \caption{The tensor product of a $K_2$ graph $uv$ and the $P_3$ graph $abc$. Here, the set of degree 1 vertices is an MEG-set.}
    \label{fig:tensor}
\end{figure}
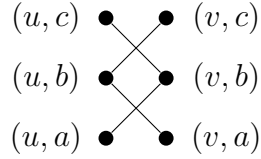

\section{Graphs of given girth}\label{sec graph parameters}

Due to Remark~\ref{rem no deg 1 in minMEG}, it makes sense to study connected graphs with minimum degree $2$. 
In Corollary~\ref{cor girth 5}, we noted that, if $G$ has girth $5$ or more, then $G$ cannot be MEG-extremal. Therefore, it is natural to wonder whether $meg(G)$ will become even smaller (with respect to the order of $G$) if $G$ becomes sparser. 
One way to consider sparse graphs  is to study graphs having high girth. The following bound is meaningful for graphs of girth at least~8. Note that, as mentioned in the introduction, an incorrect version of this bound appeared in the conference version of this paper~\cite{FMMSST24}.

\begin{theorem}\label{thm:large_girth}
   Let $G$ be a $2$-connected graph with girth $g$ at least $4$. If $G$ has $n$ vertices, then $meg(G) \leq \frac{4n}{g-3}.$
\end{theorem}

\begin{proof}
   Let $G$ be a $2$-connected graph on $n$ vertices with girth $g \geq 4$. We construct a vertex subset $M$ of 
   $G$ recursively, and claim that $M$ is an MEG-set of $G$. 
   To begin with, we initialize $M$ by picking an arbitrary vertex of $G$. Next, we add an arbitrary vertex to $M$ that is at distance at least 
   $\frac{g - 3}{4}$ from any vertex of $M$. 
   We repeat this process until every vertex of 
   $V(G) \setminus M$ are at a distance strictly less than $\frac{g - 3}{4}$ from some vertex of $M$. We want to note that $M$ can also be obtained by applying a greedy algorithm to find a maximal independent set in the $\lceil(\frac{g-3}{4})-1\rceil$-th power of $G$ when the girth is large enough.

   Next, we will show that $M$ is 
   indeed an MEG-set of $G$. Let $uv$ be an arbitrary edge of $G$. Note that, the distance between $u$ (resp., $v$) and the set $M$ is strictly less than $\frac{g-3}{4}$ due to the way we constructed $M$. 
   Without loss of generality, we consider $u'$ the nearest vertex to $u$ in $M$, such that $d_G(u, u') \leq \frac{g-3}{4}$ and $d_G(v, u')>d_G(u, u')$. If the vertex of $M$ nearest to $u$ is closer to $v$, we simply swap the place of $u$ and $v$. We consider the vertex of $M$ closest to $v$ in $G-uv$. Since $G$ is $2$-connected, $G-uv$ is connected, hence $v'$ is always defined. Also observe that $d_G(v, v') \leq 2 \left(\frac{g-3}{4}\right) - d_G(u, u')$ otherwise there would be a shortest path from $u'$ to $v'$ on more than $\frac{g-3}{4}$ vertices where at least one vertex is at distance more than $\frac{g-3}{4}$ from any vertex of $M$. 

   Let $P_1$ be a shortest path connecting $u'$ and $u$, and $P_2$ be a shortest path connecting $v$ and $v'$ in $G-uv$. Let  $P$ be the path obtained by concatenating the path $P_1$, the edge $uv$, and the path $P_2$. Since neither of $P_1$ nor $P_2$ contains the edge $uv$, $P$ is non-backtracking. Note that the length of $P$ is at most 
   $$d(u', u) + 1 + 2\left(\frac{g-3}{4}\right) - d(u', u) = \frac{g-1}{2}.$$ Moreover, no vertex can appear twice in $P$ otherwise there is a cycle of length strictly less than $g$ in $G$. Therefore, $P$ is the unique shortest path 
   connecting $u'$ and $v'$, as otherwise it will imply that there is a cycle of length 
   strictly less than $g$, contradicting the girth condition on $G$. Hence, $u'v'$ monitors the edge $uv$.

   Now we are left with counting the 
   cardinality of $M$. 
For any vertex $u \in M$, let $S_u$ be the set of all vertices that are at distance at most $\ell$ from $u$, where $\ell$ is the biggest integer such that a vertex of $G$ is at distance $\ell$ or less from at most one vertex of $M$. Notice that, for any two vertices $u,v \in M$, $S_u \cap S_v = \emptyset$, hence a path connecting $u$ and $v$ will be of length at least $2 \ell+1$. 
As any two vertices of $M$ are at distance at least $\frac{g-3}{4}$ from one another, then
$$2 \ell +1 \geq \frac{g-3}{4} \implies \ell \geq \frac{g-7}{8}.$$
As $G$ is a $2$-connected graph, each vertex $v$ of $M$ is part of a cycle which implies that $|S_v| \geq 2 \ell +1$ (counting $v$ itself too).
Therefore, we must have 
$$n = |V(G)| \geq |M|(2 \ell +1) 
\geq |M|\left(\frac{2(g-7)}{8} + 1\right)  
\implies |M| \leq \frac{4n}{g-3}.$$
This completes the proof.  
\end{proof}

\begin{remark}
    As the girth can be considered as a measure of sparseness of a graph, the above result shows that $meg(G)$ 
    has a stricter upper bound when the sparseness (in terms of the girth) of $G$ increases. However, the idea used in the proof is quite general and it may be possible to provide a better bound using the same idea for specific families of graphs, having more structural information. 
\end{remark}

\begin{theorem}\label{th girth 5 vertex cover}
    For $G$ a connected graph of minimum degree $2$ and girth at least $5$, any vertex cover of $G$ is an MEG-set.
\end{theorem}
\begin{proof}
Let $G$ be a connected graph having minimum degree at least $2$ and $n$ vertices. 
Let $I$ be an independent set. Since any vertex cover of $G$ is the complement of an independent set, it is enough to show that $M= V(G) \setminus I$ is 
an MEG-set of $G$. 

Let $uv$ be any edge of $G$. If $u, v \not\in I$, then that implies $u,v \in M$, and thus, $uv$ is monitored. If $u \in I$ and $v \not\in I$, then $v \in M$. As the minimum degree of $G$ is at least $2$, $u$ must have a neighbor $w$ (say) other than $v$. Note that as $I$ is an independent set, $w$ cannot belong to $I$. This means, in particular, $w \in M$. Notice that $wuv$ must be the unique shortest path joining $w$ and $v$ as the girth of $G$ is at least $5$, and thus $uv$ is monitored. 
\end{proof}

\begin{corollary}
    \label{cor ch number}
  Let $G$ be a connected graph with girth at least $5$ having minimum degree at least $2$. If $G$ has $n$ vertices, and chromatic number $\chi(G)$, then $meg(G) \leq n\left(\frac{\chi(G)-1}{\chi(G)}\right).$ 

\end{corollary}

\begin{proof}
    If $G$ is a connected graph with $\chi(G) = k$ and $n$ vertices, a maximum independent set of $G$ is of size at least $\frac{n}{k}$, which immediately yields the bound using Theorem~\ref{th girth 5 vertex cover}. 
\end{proof}

 We prove a more general version of Corollary \ref{cor ch number} for graphs that have pendent vertices.% in them. 

\begin{corollary}
    Let $G$ be a graph with girth at least $5$, and
    $\ell$ pendent vertices. If $G$ has $n$ vertices, and  chromatic number  $\chi(G)$, then 
    $meg(G) \leq n\left(\frac{\chi(G)-1}{\chi(G)}\right) + \frac{\ell}{\chi(G)}$.
\end{corollary}
\begin{proof}
    We know that the pendent vertices are part of any MEG-set of $G$ by Lemma~\ref{lemma:simpl}. Now, let us remove the minimum number of vertices from $G$ to obtain a subgraph $G'$ which is a connected graph having minimum degree $2$ (if possible).
    Notice that, to obtain $G'$, we need to remove
    at least the $\ell$ pendents. Therefore, $G'$ will have at most $(n-\ell)$ vertices. Therefore, using 
    Corollary~\ref{cor ch number} we can find an MEG-set of $G'$ having cardinality 
    $(n-\ell) \left(\frac{\chi(G)-1}{\chi(G)}\right)$
    as $\chi(G') \leq \chi(G)$.  Observe that the MEG-set of $G'$ together with the pendent vertices of $G$ can monitor all the edges of $G$. Therefore, 
    $$meg(G) \leq (n-\ell) \left(\frac{\chi(G)-1}{\chi(G)}\right) + \ell = n\left(\frac{\chi(G)-1}{\chi(G)}\right) + \frac{\ell}{\chi(G)}.$$

    If, in case,  it is not possible to obtain such a $G'$, then we can infer that $G$ is a forest. In that case, the set of all pendents is an MEG-set of $G$. 
    That means, $meg(G) \leq \ell$, which satisfies the statement trivially. 
\end{proof}

\section{Effects of clique-sum and subdivisions}\label{sec operations}
Let $G_1$ and $G_2$ be two graphs having cliques $C_1$ and $C_2$ of size $k$, respectively. A \textit{$k$-clique-sum} of $G_1$ and $G_2$, denoted by $G_1 \xor_k G_2$, is a graph obtained by identifying the vertices of $C_1$ with the vertices of $C_2$ (each vertex of $C_1$ is identified with exactly one vertex of $C_2$).

This particular operation between two graphs is a fundamental operation in graph theory, and is used for characterizing chordal graphs, maximal planar graphs, $K_5$-minor-free graphs, etc. Some variants of the definition requires deletion of all or some edges of the clique which may be important in the context of the problem solved. For example, in the context of the illustrious graph structure theorem~\cite{joret2013complete,robertson1999graph}, it
is allowed to delete some edges of the clique obtained by identification. In our context, as MEG-sets do not interact well with edge deletion, we investigate the changes in $meg(G)$ with respect to the clique-sum operation without edge deletion.

\begin{theorem}   
    Let $G_1 \xor_k G_2$ be a $k$-clique-sum of the graphs $G_1$ and $G_2$ for some $k \geq 2$. Then we have,  $$meg(G_1)+meg(G_2)-2k \leq meg (G_1 \xor_k G_2) \leq meg(G_1)+meg(G_2).$$
    Moreover, both the lower and the upper bounds are tight. 
\end{theorem}

\begin{proof}
Let $M_1$ and $M_2$ be MEG-sets of $G_1$ and $G_2$. Observe that, the union $M_1 \cup M_2$ is an MEG-set of $G_1 \xor_k G_2$. This implies the upper bound. 

For the lower bound, first let $C = V(G_1) \cap V(G_2)$ be the clique of size $k$, common to 
$G_1$ and $G_2$. Let $M$ be a minimum MEG-set of $G_1 
 \xor_k G_2$. Note that, $M_1 = (M \setminus V(G_2)) \cup C$  is an MEG-set of $G_1$, and $M_2 = (M \setminus V(G_1)) \cup C$  is an MEG-set of $G_2$. Observe that,
\begin{align*}
    meg(G_1) + meg(G_2) &\leq |M_1| + |M_2| \\
    &= |(M \setminus V(G_2)) \cup C| + |(M \setminus V(G_1)) \cup C| \leq |M| + 2|C| \\
    &= meg(G_1 \xor_k G_2) +2k
\end{align*}
This proves the lower bound.

\medskip

\noindent \textit{Tightness of the upper bound:}
     Take $G_1 = G_2 = K_k^*$, where $K_k^*$ denotes the graph obtained by adding a pendent neighbor to each vertex of the complete graph $K_k$. Due to Lemma~\ref{lemma:curv}, 
     $meg(G_1)=meg(G_2)=k$. 
     Note that, there is essentially a unique way to obtain a $k$-clique-sum of $G_1$ and $G_2$, and let $G = G_1 \xor_k G_2$ be the graph obtained after performing the $k$-clique-sum. Observe that $G$ is the graph obtained by adding two pendent neighbors to each vertex of the complete graph $K_k$. Therefore, by Lemma~\ref{lemma:curv}, 
     $$meg(G_1 \xor_k G_2)=2k = k + k = meg(G_1)+meg(G_2).$$

     \medskip
     
\noindent \textit{Tightness of the lower bound:}
     First we will describe the construction of a graph $H_k$. To construct this graph, we start with a
     complete graph $K_{k+1}$ on $(k+1)$ vertices named  $v_0, v_1, v_2, \cdots, v_k$. Next, we add a $3$-path of the form $v_{i,1}v_{i,2}v_{i,3}v_i$ attached to $v_i$ for each $i \in \{0,1, \cdots, k\}$. Moreover, we add edges between the vertices 
     $v_{0,2}, v_{1,2}, \cdots, v_{k,2}$ to form a clique.
     Finally we subdivide each edge of the above mentioned clique exactly once. The so obtained graph is $H_k$. We take 
     $G_1=G_2=H_k$. However, for convenience, the vertices of $G_1$ is denoted by the original names given to the vertices of $H_k$, while the vertices of $G_2$ is denoted by placing a ``bar" over the original names given to the vertices of $H_k$. That is, instead of using the names $v_0, v_1, v_{3,2}$, we will use the names $\bar{v}_0, \bar{v}_1, \bar{v}_{3,2}$ when we want to refer to the vertices of $G_2$. 
     Let $G=G_1 \xor_k G_2$ be the graph obtained by taking a $k$-clique-sum of $G_1$ and $G_2$ on the cliques $\{v_1, v_2, \cdots, v_k\}$ and 
     $\{\bar{v}_1, \bar{v}_2, \cdots, \bar{v}_k\}$. 
     
     Observe that, $H_k$ has $(k+1)$ pendents, and thus they are part of any MEG-set of $H_k$. Next suppose that $M^* = V(H_k) \setminus \{v_0, v_1, v_{0,3}, v_{1,3}\}$. Notice that, this set is not able to monitor the edge $v_0v_1$. That means, we need to take at least $k$ vertices, other than the $(k+1)$ pendent vertices, in any MEG-set of $H_k$. On the other hand, the set obtained by the pendent vertices and all but one vertex from the $K_{k+1}$ clique gives us an MEG-set of $H_k$.  Hence, $meg(H_k)=2k+1$.

     Furthermore, $G$ has exactly $(2k+2)$ pendents, and the set of all pendents of $G$ is an MEG-set. Therefore, 
     $$meg(G)=2k+2 = (2k+1)+(2k+1)-2k = meg(G_1) +meg(G_2) -2k.$$ 
     This concludes the proof. 
\end{proof}

Let $G$ be a graph. We obtain the graph $S^{\ell}_G$ by subdividing each edge of $G$ exactly $\ell$ times. 
The graph operation subdivision is also a fundamental 
graph operation, integral in the theory of topological minors, which can be used for 
sparsification of a graph. Moreover, subdivision can be considered as the inverse operation to edge contraction, which is another fundamental notion that plays an instrumental role in the famous graph minor theorem~\cite{lovasz2006graph,robertson2004graph}. 
The following result proves a relation between $meg(G)$ and $meg(S^{\ell}_G)$.

\begin{theorem}\label{thm subdivision}
    For any graph $G$ and for all $\ell \geq 2$, 
    we have $$1 \leq \frac{meg(G)}{meg(S^{\ell}_G)} \leq 2.$$ Moreover, the lower bound is tight, and the upper bound is asymptotically tight. 
\end{theorem}

\begin{proof}
Let $M$ be an MEG-set of $G$. Then observe that, $M$ is also an MEG-set of $S^{\ell}_G$. 
This proves $meg(S^{\ell}_G) \leq meg(G)$. 

On the other hand, let $M'$ be an MEG-set of $S^{\ell}_G$. 
Now we construct an MEG-set $M$ of $G$ using $M'$.
Let $v$ be a vertex of $M'$. If $v$ is also a vertex of  $G$, then we put $v$ in $M$. If $v$ is not a vertex of $G$, then $v$ must be one of the vertices which was used for subdividing an edge $e$ of $G$ to obtain $S^{\ell}_G$. 
In such a case, we put both the end points of $e$ in $M$. Note that, the so-obtained $M$ is an MEG-set of $G$. As $|M| \leq 2|M'|$, we have $meg(G)  \leq 2 meg(S^{\ell}_G) $. 
This completes the proof of the inequality. 

\medskip

\noindent \textit{Tightness of the lower bound:} Given any positive integer $n \geq 2$, take $G$ to be a tree with $n$ leaves, then $S^{\ell}_G$ is also a tree with $n$ leaves, we know that $meg(G) = meg(S^{\ell}_G) = n$~\cite{foucaud2023monitoring}. Thus, we have infinitely many examples of $G$ where, 
$$\frac{meg(G)}{meg(S^{\ell}_G)} = 1.$$ 
% Take $G = C_n$, where $C_n$ is the cycle on $n$ vertices for $n \geq 5$. Note that, in such a scenario, $S^{\ell}_G$ is the cycle $C_{n(\ell+1)}$ on $n(\ell+1)$ vertices. We know that
% $meg(G) = meg(S^{\ell}_G) = 3$~\cite{foucaud2023monitoring}. 
% Thus,  we have infinitely many examples of $G$ where, 

    \noindent \textit{Asymptotic tightness of the upper bound:}  Let $G = P_k ~\Box~ P_2$, that is the Cartesian product of the paths $P_k$ (on $k$ vertices) and $P_2$ (on $2$ vertices). For convenience, let us assume that $G$ is embedded on the plane with its vertices placed on the points $(i,j)$, where $i \in \{0, 1, \cdots, k-1\}$ and $j \in \{0,1\}$. Moreover, let the vertex placed on $(i,j)$ be denoted by $v_{i,j}$. 
    Two vertices of this graph are adjacent if they are at a Euclidean distance exactly $1$. We know that
     $meg(G) = 2k$~\cite{foucaud2023monitoring}.
     We will construct an MEG-set $M$ of size $(k+1)$ for $S^{\ell}_G$. First of all, put the vertices $v_{0,0}$ and $v_{0,1}$ in $M$. Let $w_i$ be a vertex (choose any option) on the path obtained by subdividing the edge $v_{i,0}v_{i,1}$. Put $w_i$ in $M$ for all $i \in \{1,2, \cdots, k-1\}$. Observe that, the so-obtained $M$ monitors $S^{\ell}_G$. Therefore, $meg(S^{\ell}_G)\leq k+1$.
     Hence, $$\frac{meg(G)}{meg(S^{\ell}_G)} \leq \frac{2k}{k+1} = 2 - \frac{2}{k+1}.$$
     As $k$ tends to infinity, the ratio $\frac{meg(G)}{meg(S^{\ell}_G)}$ tends to the upper bound $2$. Hence the upper bound is asymptotically tight. 
\end{proof}

\section{Concluding remarks}\label{sec conclusions}
\begin{enumerate}[(1)]
\item  In Section \ref{sec gap}, we gave examples of graphs $G_{a,b,c,d}$  which attains 
$g(G_{a,b,c,d}) = a$, 
$eg(G_{a,b,c,d}) = b$,
$seg(G_{a,b,c,d}) = c$,
and 
$meg(G_{a,b,c,d}) = d$ for ``almost" all  
$2 \leq a \leq b \leq c \leq d$. 
However, for some of the combinations of $a,b,c,d$ we still do not know if an example exists or not. One problem to consider is to decide exactly for which prescribed values of $a,b,c,d$, such a graph $G_{a,b,c,d}$ exists, along with finding an explicit example.

\item In Section~\ref{sec megfull}, we have proved a necessary and sufficient condition for a vertex to be part of any minimum MEG-set. A question in this direction is to find a necessary and sufficient condition for a vertex not to be in any minimum MEG-set. 

\item In Section~\ref{sec:ext}, we used the result of Theorem~\ref{thm meg full}, as well as a more comprehensive definition of MEG-extremal, to prove that several well-known graph classes, 
restricted to $2$-connected graphs, are MEG-extremal. We provide an example of a $2$-connected interval graph that is not MEG-extremal. Thus, it will be interesting to 
characterize or devise efficient algorithms to find $meg(G)$ when $G$ is an interval graph, or a chordal graph. One can also try to extend the results of 
Theorem~\ref{thm:extcographs} and~\ref{thm:extpropintervalgraphs} to perfect graphs.

\item In Section~\ref{sec operations}, we deal with the effects on $meg(G)$ with respect to some fundamental graph operations like clique-sums, and subdivisions. It will be interesting to perform similar studies with respect to other fundamental graph operations such as vertex deletion, edge deletion, edge contraction, graph powers, etc. One could also wonder how the parameter interacts with bounded-diameter subgraph identification, as a direct generalization of our results.

\end{enumerate}

\medskip

\noindent \textbf{Acknowledgements:} This work is partially supported 
by the following projects: IFCAM (MA/IFCAM/18/39), SERB-MATRICS (MTR/2021/000858 and MTR/2022/000692), French government IDEX-ISITE initiative 16-IDEX-0001 (CAP 20-25), International Research Center "Innovation Transportation and Production Systems" of the I-SITE CAP 20-25, and ANR project GRALMECO (ANR-21-CE48-0004), Doctoral Fellowship in India for ASEAN (DIA:2020-25).

\bibliographystyle{abbrv}
\bibliography{reference(simple)}

\end{document}